\documentclass[12pt,draftclsnofoot,onecolumn]{IEEEtran}
\usepackage{amsmath,amsthm, amssymb,graphicx,subcaption,cite}
\usepackage{braket}
\usepackage{booktabs}
\usepackage{algorithmic}
\usepackage{enumitem}
\usepackage{balance}
\usepackage{tikz}
\usetikzlibrary{positioning}
\newtheorem{remark}{Remark}

\newtheorem{theorem}{Theorem}

\usepackage{needspace}
\Needspace{5\baselineskip} % before \section to avoid awkward breaks

\usepackage{titlesec}

\titlespacing*{\section}
  {0pt}{0.8ex}{0.5ex}

\titlespacing*{\subsection}
  {0pt}{0.6ex}{0.4ex}

\titleformat{\paragraph}[runin]
  {\normalfont\normalsize\bfseries}{\theparagraph}{1em}{}[:]
\titlespacing*{\paragraph}
  {0pt}{0.4ex}{0.7em}

\makeatletter

\title{Fidelity-Age–Aware Scheduling in Quantum Repeater Networks}
\author{\IEEEauthorblockN{Ozgur Ercetin, Zafer Gedik\\}
\IEEEauthorblockA{\textit{Faculty of Engineering and Natural Sciences,} \\
\textit{Sabanci University},
Istanbul, TR. }
}

\begin{document}

\maketitle
% Alternative manual adjustment
%\vspace{-1.5\baselineskip}

\begin{abstract}
Quantum repeater networks distribute entanglement over long distances but must balance fidelity, delay, and resource contention. 
Prior work optimized throughput and end-to-end fidelity, yet little attention has been paid to the \emph{freshness} of entanglement—the time since a usable Bell pair was last delivered. 
We introduce the \emph{Fidelity-Age (FA)} metric, which measures this interval for states whose fidelity exceeds a threshold $F_{\min}$. 
A renewal formulation links slot-level success probability to long-run average FA, enabling a stochastic control problem that minimizes FA under budget and memory limits. 
Two lightweight schedulers, \emph{FA-THR} and \emph{FA-INDEX}, approximate Lyapunov-drift–optimal control. 
Simulations on slotted repeater grids show that FA-aware scheduling preserves throughput while reducing extreme-age events by up to two orders of magnitude. 
Fidelity-Age thus provides a tractable, physically grounded metric for reliable and timely entanglement delivery in quantum networks.
\end{abstract}

\section{Introduction}

Quantum repeater networks extend point-to-point quantum key distribution (QKD) links toward multi-hop entanglement distribution across large-scale topologies \cite{kimble2008quantum,wehner2018quantum}. 
Their applications—ranging from distributed quantum computation to high-rate secret key generation—depend on the timely delivery of entangled Bell pairs whose fidelities exceed a task-specific threshold.

Unlike classical packets, entangled states are transient physical resources. 
Their fidelity decays with storage time due to decoherence, and both entanglement swapping and purification introduce random delays and success probabilities \cite{briegel1998quantum,dur1999quantum}. 
As a result, the usefulness of an entanglement link depends not only on how often it is established but also on \emph{when} it was last refreshed. 
Delivering a pair that has decohered below a minimum fidelity $F_{\min}$ is equivalent to a delivery failure.

Existing work optimizes throughput and end-to-end fidelity in repeater networks \cite{pant2019routing,nickerson2014repeater,vinay2019resource,shchukin2019quantum,van2021path}, yet there is little theory on quantifying or controlling the \emph{freshness} of distributed entanglement. 
In classical communication systems, the \emph{Age of Information} (AoI) formalism provides a stochastic measure of update timeliness and has yielded provably optimal scheduling rules \cite{kaul2012real,kadota2019aoi}. 
An analogous metric for quantum networks has been missing.

This paper introduces the \emph{Fidelity-Age (FA)} metric, which measures the time elapsed since the most recent delivery of an entangled pair with fidelity at least $F_{\min}$. 
FA couples temporal freshness with physical-layer quality and directly reflects decoherence, probabilistic link generation, purification outcomes, and swapping success.

The main contributions are:
\begin{itemize}
  \item We define the FA process and derive its renewal identity, establishing a direct relation between slot-level usable-delivery probability and long-run expected FA.
  \item We formulate a stochastic control problem minimizing average FA under attempt-budget and memory constraints.
  \item We propose practical FA-aware scheduling policies—\emph{FA-THR} and \emph{FA-INDEX}—motivated by Lyapunov-drift minimization, and evaluate their performance under finite-coherence and heterogeneous-path conditions.
\end{itemize}

By linking physical fidelity evolution to a renewal-based freshness metric, this work provides a quantitative foundation for age-aware control in quantum repeater networks. 
It shows that explicit fidelity-age weighting yields schedules that are simultaneously efficient, fair, and stable under realistic decoherence limits.

\section{Background and Related Work}

% Quantum repeater networks enable the distribution of entanglement across distant nodes, forming the foundation for applications such as quantum key distribution (QKD), distributed quantum computing, and secure quantum communication. Unlike classical packets that can be duplicated and retransmitted, entangled quantum states are fragile physical resources that cannot be cloned or amplified. Their fidelity---a measure of how closely a generated state resembles an ideal Bell pair---decays over time due to decoherence in quantum memories and imperfections in physical operations. Consequently, the timely generation and use of high-fidelity entanglement are critical for reliable quantum networking.

\subsection{Qubits, Bell States, and Fidelity}

A qubit is a unit vector in $\mathbb{C}^2$, $|\psi\rangle = \alpha|0\rangle + \beta|1\rangle$ with $|\alpha|^2+|\beta|^2=1$. 
Two-qubit Bell states serve as maximally entangled targets; for example $|\Phi^+\rangle=(|00\rangle+|11\rangle)/\sqrt{2}$. 

Quantum memories store entangled states temporarily until successful swapping or purification can be performed. However, stored states decohere exponentially with time, reducing their fidelity.
For a produced two-qubit state with density matrix $\rho$, the (state) fidelity with respect to a target pure state $|\psi_{\mathrm{tar}}\rangle$ is
\begin{equation}
F \;=\; \langle \psi_{\mathrm{tar}}|\, \rho \,|\psi_{\mathrm{tar}}\rangle.
\end{equation}
We adopt the common Werner approximation for noisy pairs: write $p_W=(4F-1)/3\in[0,1]$. 
Storage for one slot maps $p_W \mapsto p_W e^{-\Delta/T_2}$, equivalently
\begin{equation}
F(\tau+1)=1-\bigl(1-F(\tau)\bigr)e^{-\Delta/T_2},
\label{eq:decoherence}
\end{equation}
and end-to-end Werner composition across a $k$-edge path yields $p_{W,\mathrm{end}}=\prod_{j=1}^k p_{W,e_j}$ and
\begin{equation}
F_{\mathrm{end}}=\frac{1+3p_{W,\mathrm{end}}}{4}.
\label{eq:swap-fidelity}
\end{equation}
When $F(\tau)$ ($F_{\mathrm{end}}$) drops below a task-specific minimum threshold $F_{\min}$, the state becomes unusable and must be discarded.

\subsection{Freshness and the Fidelity-Age Metric}

In classical networking, the \emph{Age of Information} (AoI) framework quantifies the timeliness of received updates by measuring the time elapsed since the most recent successful delivery. Inspired by this concept, we define the \emph{Fidelity-Age} (FA) metric as the time since the last delivery of an entangled pair whose fidelity exceeds $F_{\min}$. The FA process thus couples temporal freshness with physical-layer quality, reflecting both stochastic link success and fidelity evolution due to decoherence. Minimizing long-run average FA corresponds to maintaining high-quality entanglement with minimal delay between usable deliveries.

\subsection{Related Work}
Prior work on quantum repeater networks can be grouped into three themes---(i) multi-hop architectures and routing, (ii) fidelity management via purification and memory constraints, and (iii) network-level scheduling and resource allocation. We position the proposed \emph{Fidelity-Age (FA)} metric at the intersection of (ii)--(iii) by providing a freshness-oriented objective that explicitly couples stochastic delivery, fidelity thresholds, and decoherence.

%\paragraph{Architectures and routing with fidelity constraints.}
Early work established routing and multi-hop entanglement distribution models under realistic repeater operations \cite{pant2019routing}. More recent surveys clarify the physical trade-offs that link achievable rate, distance, and fidelity \cite{azuma2023review}. Building on these foundations, router-style multiplexing and end-to-end path optimization with explicit fidelity targets and memory awareness have been formulated via ILPs and scalable heuristics \cite{lee2022router,halder2024optimal}. Related protocol lines, including distance-independent-rate and packet-switching approaches, further motivate control at the network layer beyond single repeater chains \cite{patil2022distanceindependent,diadamo2022packet}.

%\paragraph{Fidelity management: purification placement, depth, and finite coherence.}
A key determinant of end-to-end usability is purification: both \emph{where} purification is performed and \emph{how many} rounds are scheduled. Recent results show that network-level scheduling of purification can substantially outperform fixed-depth designs \cite{xiao2024psc}. Separately, finite-memory and finite-coherence studies quantify rate losses due to limited storage time and highlight the need for policies that account for decoherence-induced fidelity decay \cite{semenenko2022finiteT2}. Work on reusing pre-existing entanglements likewise emphasizes that temporal state (what is already stored, and how old it is) can change path-selection decisions and reduce latency \cite{fayyaz2025pathsel}.

%\paragraph{Scheduling and resource allocation.}
At the control layer, scheduling decisions are computationally difficult (often NP-hard) and have been studied using exact formulations, heuristics, and reinforcement learning, with objectives such as throughput, fairness across requests, or utility maximization \cite{gu2023esdi,bhambay2025qswitch,ni2025dqn}. These efforts typically assume a notion of success/utility per slot, but they do not provide a unified \emph{freshness} metric that directly measures the cost of delayed usable deliveries under fidelity decay.

%\paragraph{Freshness metrics and our contribution.}
While quantum networking lacks a standardized ``age'' metric for entanglement, the classical Age-of-Information (AoI) literature offers a principled framework for freshness-aware optimization, including Lyapunov-drift methods \cite{kadota2019aoi}. We adapt this viewpoint to repeater networks by defining \emph{Fidelity-Age (FA)}: successful deliveries with end-to-end fidelity at least $F_{\min}$ serve as renewal events, and the time since the most recent such delivery quantifies freshness-conditioned usability. Relative to routing and path-optimization work \cite{pant2019routing,lee2022router,halder2024optimal}, purification scheduling \cite{xiao2024psc}, and general scheduling approaches \cite{gu2023esdi,bhambay2025qswitch,ni2025dqn,fayyaz2025pathsel}, our focus is an objective and analysis framework that integrates (a) stochastic entanglement generation and swapping, (b) fidelity thresholds, and (c) decoherence over finite coherence times, yielding scheduling rules explicitly targeted at freshness-conditioned usability.

\section{System Model}
\label{sec:SystemModel}
We formalize a slotted quantum-repeater network on an explicit probability space and specify storage, decoherence, purification, and swapping so that the model is consistent with the numerical experiments.

\subsection{Network Architecture}
\label{sec:arch}

A software-defined network (SDN) controller maintains a global view of link states and stored pairs, issues per-slot commands for link attempts, swapping, and optional purification, and uses the classical control plane to coordinate quantum operations.
Each intermediate repeater holds two quantum memories to serve its two incident links in a line topology; end nodes hold one memory per active end-to-end demand. When entanglement diversity (parallel path attempts or later distillation) is enabled, end nodes allocate an extra memory unit per additional concurrent end-to-end (e2e) attempt.
A \emph{path} is a simple sequence of physical links connecting a source--destination pair. The controller may reuse \emph{prior entanglements} left from previous cycles on any subset of path edges. Reuse can tilt decisions toward longer paths that already hold stored pairs.
Unless stated otherwise, hardware is uniform: equal inter-nodal distances, identical sources and BSM modules.

\subsection{Probabilistic Slotted Graph Model}
Let $G=(V,E)$ be a connected undirected graph of repeaters and lossy optical links. Time is slotted with fixed slot duration $\Delta>0$; slots are indexed by $t\in\mathbb{N}$. The system is defined on $(\Omega,\mathcal{F},\mathbb{P})$.

Each edge $e\in E$ has $S(e)\in\mathbb{N}$ parallel modes. In slot $t$, mode $i\in\{1,\dots,S(e)\}$ attempts a link-level entanglement with success indicator
\begin{align*}
X_t^{(i)}(e)&\sim \mathrm{Bernoulli}\big(p_0(e)\big),\\
p_0(e)&=\eta(e),\qquad \eta(e)=\exp(-\alpha L(e)),
\end{align*}
where $L(e)$ is the link length and $\alpha>0$ is the fiber loss coefficient (km$^{-1}$)\cite{Nielsen_Chuang_2010}.

\begin{remark}[Independence]
\label{rem:indep}
The random variables $\{X_t^{(i)}(e)\}$ are independent across $e$, $i$, and $t$, and independent of all Bell-state measurement (BSM) and purification outcomes.
\end{remark}

The number of successful link-level pairs on $e$ in slot $t$ is
\[
N_t(e)=\sum_{i=1}^{S(e)}X_t^{(i)}(e)\sim\mathrm{Binomial}\!\big(S(e),p_0(e)\big).
\]
The per-edge success probability (at least one success on $e$) is
\[
p_{\text{link}}(e)=1-\big(1-p_0(e)\big)^{S(e)}.
\]

We cap the number of retained pairs per edge per slot by $m_{\max}(e)\in\mathbb{N}$:
\[
N_t^{\mathrm{ret}}(e)=\min\{N_t(e),\,m_{\max}(e)\}.
\]
Note that $m_{\max}(e)=1$ when there is no purification, whereas $m_{\max}(e)\ge 2$ when there is.

We also define the \emph{external multigraph} $\mathcal{H}_t^{\mathrm{multi}}$ whose edge multiplicity on $e$ equals $N_t^{\mathrm{ret}}(e)$, and its simple support
\[
\mathcal{H}_t=\{e\in E:\,N_t^{\mathrm{ret}}(e)\ge 1\}.
\]

\subsection{Internal phase: purification and swapping}
Internal operations are local to repeaters and occur after the external phase within the same slot.
We adopt the standard BBPSSW entanglement purification protocol with a fixed noise model (bit-flip or Werner) throughout the experiments; see \cite{Bennett1996Purif} for the explicit maps and success probabilities. In all cases, input-pair fidelities already include storage decay via~\eqref{eq:decoherence}.

We assume independent BSMs at intermediate repeaters with per-repeater success probability $q\in(0,1]$. For a path with $k$ edges, a swap attempt succeeds with probability $q^{\,k-1}$, conditional on one available pair per edge and independence across repeaters \cite{Nielsen_Chuang_2010}.

Edge pairs are modeled as Werner states. Conditioned on successful swaps, the end-to-end state remains Werner with a parameter equal to the product of edge parameters; we use this only to check usability $F_{\mathrm{end}}(\pi)\ge F_{\min}$ and omit the standard closed forms \cite{Nielsen_Chuang_2010}. For $T_{\mathrm{mem}}>1$, edge fidelities already include storage decay and any in-slot purification.

\subsection{Usable deliveries and success events}
Fix a set $\mathcal{P}$ of source--destination requests.
An application-level threshold $F_{\min}\in(1/4,1)$ defines \emph{usable} deliveries.
In slot $t$, for flow $(s,d)$ we say a usable delivery occurs if there exists a simple path $\pi\subseteq \mathcal{H}_t$ from $s$ to $d$ with $k:=|\pi|$ edges such that:
\begin{enumerate}
\item Each edge $e\in\pi$ has at least one pair available for swapping after optional purification (and storage decay when $T_{\mathrm{mem}}>1$);
\item All $k-1$ BSMs on $\pi$ succeed (probability $q^{\,k-1}$);
\item The resulting end-to-end fidelity satisfies $F_{\text{end}}(\pi)\ge F_{\min}$.
\end{enumerate}
Let $Y_t^{(s,d)}\in\{0,1\}$ be the indicator of this event.
If multiple edge-disjoint paths are activated in the same slot, we attempt them independently; independence across edge-disjoint paths follows from Remark~\ref{rem:indep} and disjoint resource usage.
The per-slot throughput is $N_{\mathrm{succ}}(t)=\sum_{(s,d)\in\mathcal{P}}Y_t^{(s,d)}$.

Although path feasibility is described in terms of the simple support graph $\mathcal H_t$, simultaneous path activations are constrained by the retained-pair multiplicities $N_t^{\mathrm{ret}}(e)$. Each stored entangled pair can be consumed by at most one swapping operation in a slot. Consequently, if two candidate paths (possibly serving different source--destination pairs) share an edge $e$, they may be activated in the same slot only if $N_t^{\mathrm{ret}}(e)\ge 2$; when $N_t^{\mathrm{ret}}(e)=1$, at most one such path can be scheduled. This exclusivity is enforced by the controller action $a(t)$.

\subsection{Controller Actions and Constraints}
\label{subsec:controller}

At each slot $t$, the controller observes the realized external successes $\{N_t^{\mathrm{ret}}(e)\}$ defining $\mathcal{H}_t^{\mathrm{multi}}$ and the fidelity--age states of all stored pairs.
It then selects an action
\begin{align*}
a(t)\in\mathcal A_t
=\{\text{choose up to }R\text{ flow--path activations and optional purifications}\},
\end{align*}
subject to the following constraints:
\begin{itemize}[leftmargin=*]
\item \textbf{Attempt budget:} at most $R$ activations per slot;
\item \textbf{Local exclusivity:} each stored pair is used at most once; purification consumes two pairs on the same link;
\item \textbf{Memory:} a node stores at most $m_{\max}(e)$ pairs per edge, each expiring after $T_{\mathrm{mem}}$ slots under the decoherence law~\eqref{eq:decoherence}.
\end{itemize}
No cross-slot coordination or global optimization is assumed; the resulting dynamics are Markovian.

\subsection{Homogeneous and heterogeneous regimes}
Let a simple path $\pi=(e_1,\dots,e_k)$ comprise $k$ edges.
Under homogeneous assumptions
\[
S(e)\equiv S,\quad p_0(e)\equiv p_0,\quad q\ \text{constant},\quad F_0(e)\equiv F_0,
\]
each edge succeeds independently with probability
\[
p_{\text{link}} = 1 - (1-p_0)^S.
\]
Hence, by independence,
\[
\Pr\{\pi \subseteq \mathcal H_t\}
= \prod_{j=1}^{k}\Pr\{e_j\in\mathcal H_t\}
= (p_{\text{link}})^{k}.
\]
Conditional on this event, the $k-1$ intermediate BSMs succeed independently with probability $q^{\,k-1}$.
Therefore, the unconditional probability that path $\pi$ yields one end-to-end ebit in slot~$t$ is
\[
P_{\mathrm{e2e}}(\pi)=(p_{\text{link}})^{k} q^{\,k-1}.
\]
Numerical results also consider \emph{heterogeneous} links with $L(e)$ drawn from an interval, leading to edge-dependent $p_0(e)$ and baseline fidelities $F_0(e)$ (via hardware- or distance-dependent calibration).
In all heterogeneous experiments we keep~\eqref{eq:swap-fidelity} for fidelity composition and~\eqref{eq:decoherence} for storage decay.

\begin{remark}[Stationarity]
\label{rem:stationary}
Under Remark~\ref{rem:indep} and fixed parameters, the external multigraph process $\{\mathcal{H}_t^{\mathrm{multi}}\}$ is i.i.d. across slots; with $T_{\mathrm{mem}}>1$ the joint state (stored-pair inventories with ages) is a time-homogeneous Markov chain. These properties are the basis for renewal and steady-state analyses used later.
\end{remark}

\section{Problem Statement}
\label{sec:ProblemFormulation}
This section formalizes the Fidelity-Age (FA) control problem by defining the usable-delivery process, the resulting FA dynamics, and the optimization objective under the system constraints introduced in Sec.~\ref{sec:SystemModel}.

\subsection{Usable-delivery process}
\label{sec:usable}
From Sec.~\ref{sec:SystemModel}, in each slot $t$ the controller observes the random multigraph $\mathcal H_t^{\mathrm{multi}}$ of link-level successes and executes internal operations according to its chosen action $a(t)$.
Let $Y_t^{(s,d)}\in\{0,1\}$ denote the event that request $(s,d)\in\mathcal P$ receives a \emph{usable} entangled pair in slot~$t$, i.e.,
\begin{align*}
Y_t^{(s,d)}
=\mathbf 1\!&\left\{
\exists\, \pi\subseteq\mathcal H_t \text{ selected by } a(t):\text{$\pi$ connects $(s,d)$,}\right.\\
&\left.
\text{ all swaps succeed, and } F_{\text{end}}(\pi)\ge F_{\min}
\right\}.
\end{align*}
The selection by $a(t)$ ensures that no physical entangled pair is used by more than one path in the same slot. Usable deliveries depend on the realized external successes, any applied purifications, storage decay, and the controller’s scheduling choice.

\subsection{Fidelity-Age (FA)}
For each request $(s,d)$ define the sequence of usable-delivery epochs
\[
t_1^{(s,d)} < t_2^{(s,d)} < \cdots,
\qquad
t_{n+1}^{(s,d)} = \inf\{t>t_n^{(s,d)} : Y_t^{(s,d)}=1\}.
\]
The \emph{Fidelity-Age} process measures the number of slots since the most recent usable delivery,
\begin{equation}
A_{s,d}(t) = t - \max\{\,t_n^{(s,d)} \le t\,\},
\label{eq:fa-def}
\end{equation}
with $A_{s,d}(t)=t$ if no usable delivery has yet occurred.
$A_{s,d}(t)$ increases linearly between deliveries and resets to~0 whenever $Y_t^{(s,d)}=1$.

\subsection{Optimization objective}
Controller decisions obey the per-slot constraints defined in Sec.~\ref{subsec:controller}.
The performance metric of interest is the long-run time-average Fidelity-Age,
\begin{equation}
\overline A_{s,d}
=\limsup_{n\to\infty}\frac{1}{n}\sum_{t=1}^{n} A_{s,d}(t),
\end{equation}
and the system objective is to minimize a weighted sum of these averages:
\begin{equation}
\label{eq:fa-opt}
\min_{\phi\in\Phi}
\sum_{(s,d)\in\mathcal P} w_{s,d}\,\overline A_{s,d},
\quad
\text{s.t. feasibility and budget constraints,}
\end{equation}
where $\Phi$ is the set of stationary nonanticipative control policies and $\{w_{s,d}\}$ are fixed nonnegative weights.

\subsection{Slot-level success probability}
For a path $\pi=(e_1,\ldots,e_k)$ activated in slot $t$, the unconditional probability of a usable delivery is
\begin{equation}
\label{eq:path-success}
P_{\mathrm{use}}(\pi)
=\Biggl(\prod_{j=1}^{k} p_{\text{link}}(e_j)\Biggr)\,
q^{\,k-1}\,
\mathbf 1\!\left\{F_{\text{end}}(\pi)\ge F_{\min}\right\}.
\end{equation}
In the homogeneous regime, $p_{\text{link}}(e)\equiv p_{\text{link}}$ and $P_{\mathrm{use}}(\pi)=p_{\text{link}}^{k}q^{\,k-1}$ whenever $F_{\text{end}}(\pi)\ge F_{\min}$.

\subsection{Renewal interpretation}
Under Remark~\ref{rem:indep} and any stationary policy $\phi\in\Phi$, the joint process $\{Z_t\}$ of system states is a time-homogeneous Markov chain.
If for every $(s,d)$ the inter-delivery time $\tau^{(s,d)}=t_{n+1}^{(s,d)}-t_n^{(s,d)}$ has finite second moment, then standard renewal theory yields
\begin{equation}
\label{eq:renewal-identity}
\overline A_{s,d}
=\frac{\mathbb E[(\tau^{(s,d)})^2]}{2\,\mathbb E[\tau^{(s,d)}]}.
\end{equation}
This identity connects slot-level success probabilities in~\eqref{eq:path-success} to long-run FA performance and underlies the policy analysis in Sec.~\ref{sec:analysis}.

\section{Analysis}
\label{sec:analysis}

We analyze the Fidelity-Age (FA) process under a stationary control policy $\phi$ using renewal theory on the slotted quantum-network model.

\subsection{Renewal formulation}
Let $Y_t$ denote the single-flow usable-slot indicator from Sec.~\ref{sec:usable} (we drop $(s,d)$ for brevity).
Let the inter-delivery interval be
\begin{equation}
\tau=\inf\{\ell\ge 1:\,Y_{t+\ell}=1\mid Y_t=1\},
\label{eq:tau}
\end{equation}
measured in slots of duration $\Delta$. The renewal epochs $\{t_k\}$ correspond to times when a usable Bell pair is delivered.

\begin{theorem}[FA renewal identity]
\label{thm:renewal}
Fix a slot duration $\Delta>0$ and a stationary policy $\phi$.
Let $Y_t\in\{0,1\}$ be the usable-delivery indicator from Sec.~\ref{sec:usable} and let
$\tau=\inf\{\ell\ge 1:\,Y_{t+\ell}=1\mid Y_t=1\}$ be the inter-delivery time (in slots).
Assume the underlying state process
\[
Z_t=(\mathcal H_t^{\mathrm{multi}},\,\text{stored-pair inventories and ages},\,a(t))
\]
is stationary and ergodic under $\phi$, that $0<\Pr\{Y_t=1\}<1$, and that
$\mathbb{E}[\tau]<\infty$ and $\mathbb{E}[\tau^2]<\infty$.
Then the long-run average Fidelity--Age for flow $(s,d)$ exists and
\begin{equation}
\bar A_{s,d}
=\frac{\mathbb{E}[(\Delta \tau)^2]}{2\,\mathbb{E}[\Delta \tau]}
=\Delta\,\frac{\mathbb{E}[\tau^2]}{2\,\mathbb{E}[\tau]}\,.
\label{eq:FA-renewal}
\end{equation}
\end{theorem}

\begin{proof}[Proof sketch]
Let $\{\sigma_n\}$ be delivery epochs (in continuous time) with inter-delivery lengths $\tau_n$ (in slots), i.e., $\sigma_{n+1}-\sigma_n=\Delta \tau_n$.
Between $\sigma_n$ and $\sigma_{n+1}$ the age grows linearly from $0$ to $\Delta \tau_n$ (continuous evolution within slots). The cycle reward is the time--integral of age,
\[
R_n=\int_{\sigma_n}^{\sigma_{n+1}} A_{s,d}(u)\,du
=\tfrac{1}{2}(\Delta \tau_n)^2
=\tfrac{\Delta^2}{2}\,\tau_n^2.
\]
By the renewal--reward theorem, under the stated assumptions,
\[
\bar A_{s,d}
=\lim_{t\to\infty}\frac{\sum_{n:\,\sigma_n<t} R_n}{t}
=\frac{\mathbb E[R]}{\mathbb E[\Delta \tau]}
=\Delta\,\frac{\mathbb E[\tau^2]}{2\,\mathbb E[\tau]},
\]
which yields~\eqref{eq:FA-renewal}.
\end{proof}

\subsection{Closed-form geometric case}
If the usable-slot success probability is constant, $p_{\mathrm{succ}}=\Pr\{Y_t=1\}$, then
$\tau\sim \mathrm{Geom}(p_{\mathrm{succ}}),$
and
\begin{equation}
\bar A_{s,d}=\tfrac{\Delta}{2}\,\tfrac{2-p_{\mathrm{succ}}}{p_{\mathrm{succ}}}.
\label{eq:geom-FA}
\end{equation}
This expression defines the FA--success trade-off and is used for simulation baseline validation.
The derivative $\partial\bar A_{s,d}/\partial p_{\mathrm{succ}}=-\Delta/(2p_{\mathrm{succ}}^2)$ confirms that FA decreases monotonically with $p_{\mathrm{succ}}$ but with diminishing benefit as $p_{\mathrm{succ}}\to 1$.

\subsection{Slot-level success probability}
Consider homogeneous links with per-slot external success $p_{\text{link}}$ and swap reliability $q$.
A path $\pi$ of $k$ edges is usable with probability
\begin{equation}
P_{\mathrm{use}}(\pi)=p_{\text{link}}^{k}q^{k-1}\,\mathbf1\{F_{\mathrm{end}}(\pi)\ge F_{\min}\}.
\label{eq:spath}
\end{equation}
If one purification round precedes each swap, with success map $P_{\mathrm{pur}}(F_1,F_2)$ and output fidelities $F'_\ell$, then
\begin{align}
P_{\mathrm{use}}(\pi)
=\mathbb E\!\left[
\mathbf1\{\pi\subseteq\mathcal H_t\}
\right.&\prod_{\ell\in \pi} P_{\mathrm{pur}}(F_{1,\ell},F_{2,\ell})
\,q^{k-1}\,\notag\\
&\left.\mathbf1\{F_{\mathrm{end}}(\{F'_\ell\})\ge F_{\min}\}
\right].
\label{eq:spath-pur}
\end{align}
When $\kappa$ paths $\pi_1,\dots,\pi_\kappa$ are activated such that no physical entangled pair is shared across paths (i.e., they are disjoint in the multigraph $\mathcal H_t^{\mathrm{multi}}$), independence implies
\[
p_{\mathrm{succ}}
=1-\prod_{i=1}^{\kappa}\bigl(1-P_{\mathrm{use}}(\pi_i)\bigr)
\approx \sum_{i=1}^{\kappa} P_{\mathrm{use}}(\pi_i)
\]
for small $P_{\mathrm{use}}(\pi_i)$, showing that parallelism improves the slot-level success probability in the low-success regime.

\subsection{Finite-FA conditions and stability}
Finite $\bar A_{s,d}$ requires $\mathbb E[\tau^2]<\infty$.
This holds whenever per-link and per-swap success probabilities are bounded below by $p_->0$ and $q_->0$, and the policy attempts at least one finite-length path with positive probability.
Then $\tau$ is stochastically dominated by a geometric variable, ensuring $\bar A_{s,d}<\infty$.
These conditions cover all simulation settings in Section~\ref{sec:numerical_results}.

\subsection{Control sensitivity and scheduling}
Combining~\eqref{eq:FA-renewal}--\eqref{eq:spath-pur}, any action that increases the conditional success probability
$\Pr\{Y_t=1\mid \text{state at }t,\,a(t)\}$ reduces FA.
Hence, the controller should maximize the expected slot-level success given the observed state:
\begin{equation}
a^*(t)=\arg\max_{a\in\mathcal A_t}\mathbb E\!\left[Y_t\,\middle|\,\text{state at }t,\,a\right],
\label{eq:rule}
\end{equation}
subject to the per-slot attempt and memory limits.
Equation~\eqref{eq:rule} represents a drift-minimizing heuristic analogous to Lyapunov optimization in queueing networks and motivates the FA-THR and FA-INDEX schedulers analyzed later.

\section{Numerical Results}
\label{sec:numerical_results}
Monte Carlo simulations were carried out using the probabilistic slotted model described in Section~III.
Each experiment covers $2{\times}10^5$ slots with a $2{\times}10^4$-slot warm-up and five random seeds.
The network is a $3{\times}3$ grid supporting $|P|{=}16$ source–destination pairs with an attempt budget $R{=}8$ per slot ($R/|P|{=}0.5$).
Unless otherwise stated, parameters are $\alpha{=}0.046~\mathrm{km}^{-1}$, $S{=}8$, BSM success $q{=}0.95$, $F_\mathrm{base}{=}0.82$, and $F_{\min}{=}0.75$.

\subsection{Algorithms and Benchmarks}
\label{subsec:algorithms}

We evaluate four schedulers: two conventional baselines and two FA-aware designs. 
At each slot $t$, a controller selects a subset $\mathcal S_t\!\subseteq\!\mathcal P$ of flows to serve, under the per-slot budget $|\mathcal S_t|\!\le\!R$.

\textbf{Conventional baselines.}
\emph{TP-MAX} chooses $\mathcal S_t$ maximizing expected successful entanglements 
$\sum_{(s,d)\in\mathcal S_t}\hat p_{sd}(t)$, 
where $\hat p_{sd}(t)$ is the estimated success probability.  
\emph{FID-MAX} replaces $\hat p_{sd}(t)$ with the expected end-to-end fidelity $\hat F_{sd}(t)$.
Both greedily optimize instantaneous utility and ignore flow-age evolution.

\textbf{FA-aware schedulers.}
Flow-age–aware control incorporates the penalty $A_{sd}(t)$ into the Lyapunov-drift rule
\begin{equation}
\mathcal S_t^*=\arg\max_{\mathcal S\subseteq\mathcal P,\,|\mathcal S|\le R}
\sum_{(s,d)\in\mathcal S}\!\big[U_{sd}(t)-\gamma A_{sd}(t)\big],
\label{eq:fa_rule}
\end{equation}
where $U_{sd}(t)$ is instantaneous utility and $\gamma>0$ a sensitivity weight.  
We implement two low-complexity approximations:

\emph{FA-THR}$(\tau)$ activates flows with $A_{sd}(t)>\tau$ and serves the $R$ with highest $U_{sd}(t)$.  
Small $\tau$ favors throughput (TP-MAX limit), large $\tau$ promotes fairness.

\emph{FA-INDEX} ranks flows by
\begin{equation}
I_{sd}(t)=U_{sd}(t)/(1+\beta A_{sd}(t)),
\end{equation}
serving the $R$ largest indices. 
The parameter $\beta>0$ controls the age–throughput trade-off; $\beta\!\to\!0$ yields TP-MAX, and large $\beta$ approaches equal-age service. 
Values $\beta\!\in\![0.05,0.2]$ are empirically stable.

\textbf{Complexity.}
All schedulers require sorting at most $|\mathcal P|$ entries per slot:
$\mathcal O(|\mathcal P|\log|\mathcal P|)$.  
FA-THR adds one comparison, and FA-INDEX adds a scalar weight, negligible for $|\mathcal P|\!\le\!16$.  
No global optimization or iterative search is needed.

\textbf{Interpretation.}
Both FA-THR and FA-INDEX operationalize~\eqref{eq:fa_rule} by weighting utility with age. 
FA-THR provides a coarse serve/delay decision, while FA-INDEX yields smooth prioritization. 
Together they span the fairness–efficiency envelope of the Lyapunov-optimal policy.

\subsection{Performance Metrics}
\label{subsec:metrics}

Each simulation produces per-flow age traces $A_{sd}(t)$, the number of slots since the last usable entanglement for pair $(s,d)$.
From these traces we compute the following metrics.

\textbf{Mean Age.}
\begin{equation}
\bar A_{sd}=\frac{1}{T}\sum_{t=1}^{T}A_{sd}(t),\qquad
\bar A=\frac{1}{|\mathcal P|}\!\sum_{(s,d)\in\mathcal P}\bar A_{sd}.
\label{eq:Amean}
\end{equation}

\textbf{Tail Age ($A_{95}$).}
\begin{equation}
A_{95}=\inf\{a:\Pr[A_{sd}(t)\le a]\ge0.95\}.
\label{eq:A95}
\end{equation}

\textbf{Extreme-Age Risk ($\mathrm{CVaR}_{95}$).}
\begin{equation}
\mathrm{CVaR}_{95}
=\mathbb E[A_{sd}(t)\mid A_{sd}(t)>A_{95}].
\label{eq:CVaR95}
\end{equation}

\textbf{Throughput.}
\begin{equation}
\mathrm{Thr}=\frac{1}{T}\sum_{t=1}^{T}N_{\mathrm{succ}}(t),
\end{equation}
where $N_{\mathrm{succ}}(t)$ counts usable deliveries with $F_{sd}(t)\!\ge\!F_{\min}$.

\textbf{Fairness.}
Jain’s index on inverse ages serves as a service-rate proxy:
\begin{equation}
J=
\frac{\big(\sum_{(s,d)}(1+\bar A_{sd})^{-1}\big)^2}
{|\mathcal P|\sum_{(s,d)}(1+\bar A_{sd})^{-2}}.
\label{eq:Jain}
\end{equation}

\textbf{Starvation.}
\begin{equation}
S=\frac{1}{|\mathcal P|}
\sum_{(s,d)}\mathbf1\{\bar A_{sd}>A_{\mathrm{ref}}\},
\label{eq:Starv}
\end{equation}
where $A_{\mathrm{ref}}$ is the 95th-percentile age of a single-flow baseline.

All metrics include 95\% confidence intervals from batch-mean estimates across Monte-Carlo seeds and are referenced in Figs.~\ref{fig:fa_cdf_policy}–\ref{fig:fa_frontier} and Table~\ref{tab:load_results}.

\subsection{Baseline Validation}
A single-flow configuration ($|P|{=}1$, $R{=}1$) yields a mean flow age $\bar A{=}0.09$ and throughput $s{\approx}0.918$ deliveries per slot.
The close match between the simulated $\bar A$ and the geometric prediction $\tfrac{1}{2}\tfrac{2-s}{s}$ validates both the stochastic event scheduler and the renewal approximation that underpins our analytical derivations.
This agreement confirms that deviations observed in the multi-user setting stem from scheduling dynamics rather than model artifacts.

\subsection{Load and Contention Effects}
When multiple flows contend for limited transmission attempts, TP-MAX and FID-MAX greedily maximize instantaneous utility (throughput or fidelity) at the expense of temporal balance.
This myopic prioritization allows a small subset of flows to monopolize resources, while others remain idle for extended periods.
The outcome is a heavy-tailed flow-age distribution and catastrophic latency growth ($\bar A{\approx}2.75{\times}10^4$, $A_{95}{\approx}5.5{\times}10^4$) with roughly $25\%$ starvation, as shown in Table~\ref{tab:load_results}.

FA-aware schedulers, in contrast, internalize this temporal imbalance by explicitly penalizing aged flows.
They occasionally defer high-utility transmissions to rejuvenate stale flows, thereby keeping the system near a quasi-stationary equilibrium.
This adaptive balancing substantially improves fairness and stability: the cumulative distributions in Fig.~\ref{fig:fa_cdf_policy} show that FA-aware policies dominate the tail region, eliminating the extreme delays characteristic of throughput-driven baselines.
The resulting network state exhibits near-equalized service across flows, reflected in Jain’s index $J{\approx}0.91$ and negligible starvation for FA-INDEX.

Figure~\ref{fig:fa_frontier} presents the throughput–tail-age frontier, comparing throughput against the 95\% conditional value-at-risk (CVaR) of flow age.
The behavior of FA-aware schedulers can be interpreted as a heuristic alignment with the control rule in~\eqref{eq:rule}.
Since the long-run mean flow age $\bar A_{sd}$ decreases monotonically with the usable-slot success probability $s$, 
Eq.~\eqref{eq:rule} prescribes allocating transmission attempts to actions that most increase $\mathbb{E}[s(a(t))]$ per slot.
FA-INDEX and FA-THR operationalize this idea by weighting each flow’s instantaneous utility by a decreasing function of its current age, effectively biasing service toward links whose success would yield the largest marginal reduction in expected age.
Although these policies do not explicitly minimize a Lyapunov drift, their decisions empirically follow the same stabilizing direction, keeping the network near a balanced operating point with bounded flow ages.

\begin{figure}[t]
\centering
\begin{subfigure}[t]{0.48\linewidth}
  \centering
  \includegraphics[width=\linewidth]{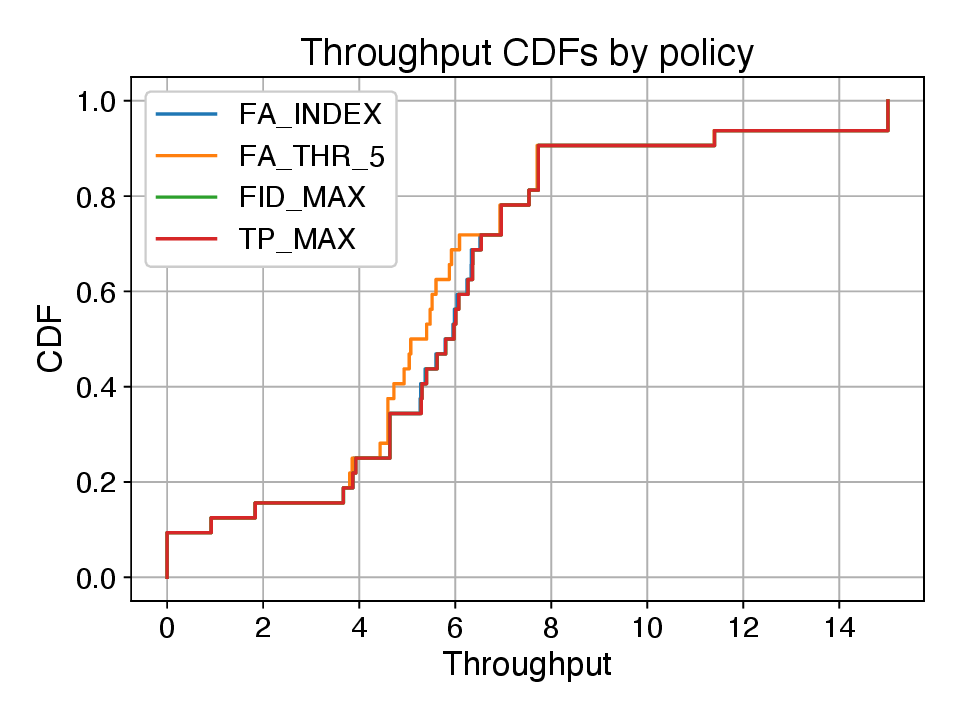}
  \caption{\footnotesize
  CDF of per-flow ages $A_{sd}(t)$ showing the empirical distribution of flow-age samples used in
  Eqs.~\eqref{eq:Amean}–\eqref{eq:CVaR95}.
  FA-aware schedulers (FA-THR, FA-INDEX) exhibit sharply bounded tails,
  whereas throughput-driven baselines (TP-MAX, FID-MAX) produce heavy-tailed delay and starvation.}
  \label{fig:fa_cdf_policy}
\end{subfigure}
\hfill
\begin{subfigure}[t]{0.48\linewidth}
  \centering
  \includegraphics[width=\linewidth]{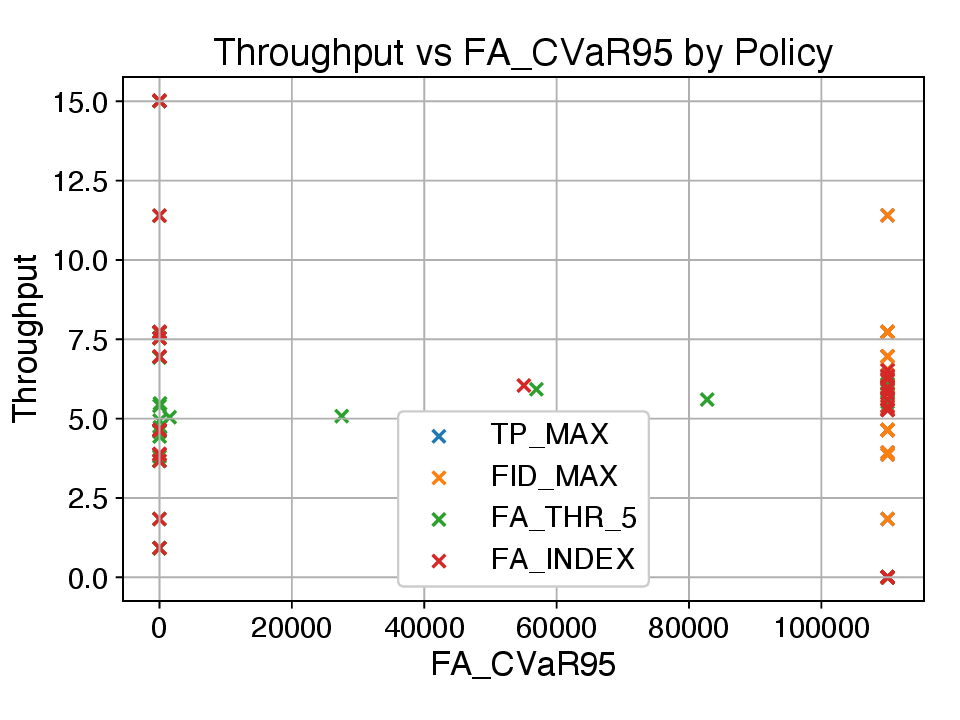}
  \caption{\footnotesize
  Scatter plot of mean throughput (Eq.~\eqref{eq:Amean}) versus the 95\% conditional value-at-risk of flow age ($\mathrm{CVaR}_{95}$, Eq.~\eqref{eq:CVaR95}).
  Each point corresponds to one scheduler.
  FA-aware policies achieve Pareto-efficient trade-offs, maintaining throughput comparable to TP-MAX while reducing extreme-age risk by two orders of magnitude.}
  \label{fig:fa_frontier}
\end{subfigure}
\vspace{1ex}
\caption{
$3{\times}3$ grid topology, $|P|{=}16$ source–destination pairs, attempt budget $R{=}8$ per slot ($R/|P|{=}0.5$),
$q{=}0.95$, $F_{\min}{=}0.75$, no purification.
Comparison of flow-age distributions and throughput–tail-age frontiers across scheduling policies.}
\label{fig:B_dual}
\end{figure}

\begin{table}[t]
\centering
\caption{Load contention on $3\times3$ grid ($|P|{=}16$, $R/|P|{=}0.5$, $F_{\min}{=}0.75$, no purification).}
\label{tab:load_results}
\renewcommand{\arraystretch}{1.15}
\setlength{\tabcolsep}{4.5pt}
\begin{tabular}{lccccc}
\toprule
\textbf{Policy} & $\bar A$ & $A_{95}$ & Throughput & Jain & Starv. \\
\midrule
TP-MAX    & $2.75{\times}10^{4}$ & $5.5{\times}10^{4}$ & $6.61$ & $0.75$ & $0.25$ \\
FID-MAX   & $2.75{\times}10^{4}$ & $5.5{\times}10^{4}$ & $6.61$ & $0.75$ & $0.25$ \\
FA-THR(5) & $0.77$ & $1.34$ & $6.60$ & $0.89$ & $0.00$ \\
FA-INDEX  & $0.67$ & $2.11$ & $6.61$ & $0.91$ & $0.01$ \\
\bottomrule
\end{tabular}
\end{table}

% \begin{figure}[t]
% \centering
% \includegraphics[width=0.9\linewidth]{figs/B/B_FA_p95.pdf}
% \caption{Flow-age percentiles across $(|P|,R/|P|)$ configurations. 
% FA-aware schedulers maintain bounded tails while throughput-oriented policies diverge.}
% \label{fig:B_fa_p95}
% \end{figure}

% \begin{figure}[t]
% \centering
% \includegraphics[width=0.9\linewidth]{figs/B/B_frontier_thr_vs_FA95.pdf}
% \caption{Throughput--FA$_{95}$ trade-off frontier. 
% FA-aware schedulers dominate the Pareto region.}
% \label{fig:B_frontier}
% \end{figure}

\begin{figure}[t]
\centering
\begin{subfigure}[t]{0.48\linewidth}
  \centering
  \includegraphics[width=\linewidth]{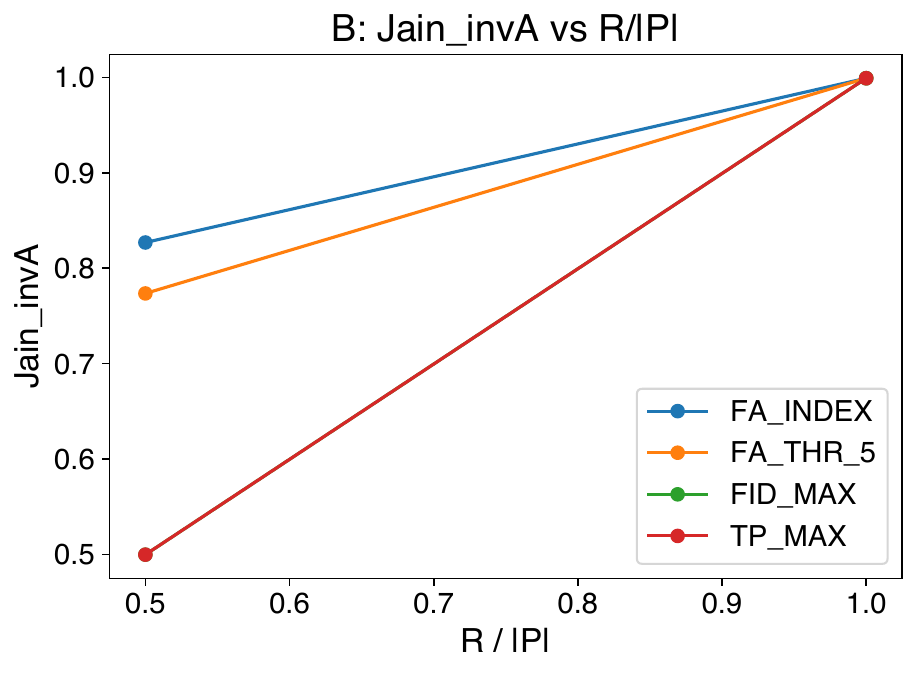}
  \caption{\footnotesize
  Fairness measured by Jain’s index $J$ (Eq.~\eqref{eq:Jain}) computed on inverse mean ages $1/(1+\bar A_{sd})$.
  FA-aware policies (FA-THR, FA-INDEX) sustain high fairness as the offered load $R/|P|$ increases, 
  whereas throughput-oriented baselines collapse under contention.}
  \label{fig:B_jain}
\end{subfigure}
\hfill
\begin{subfigure}[t]{0.48\linewidth}
  \centering
  \includegraphics[width=\linewidth]{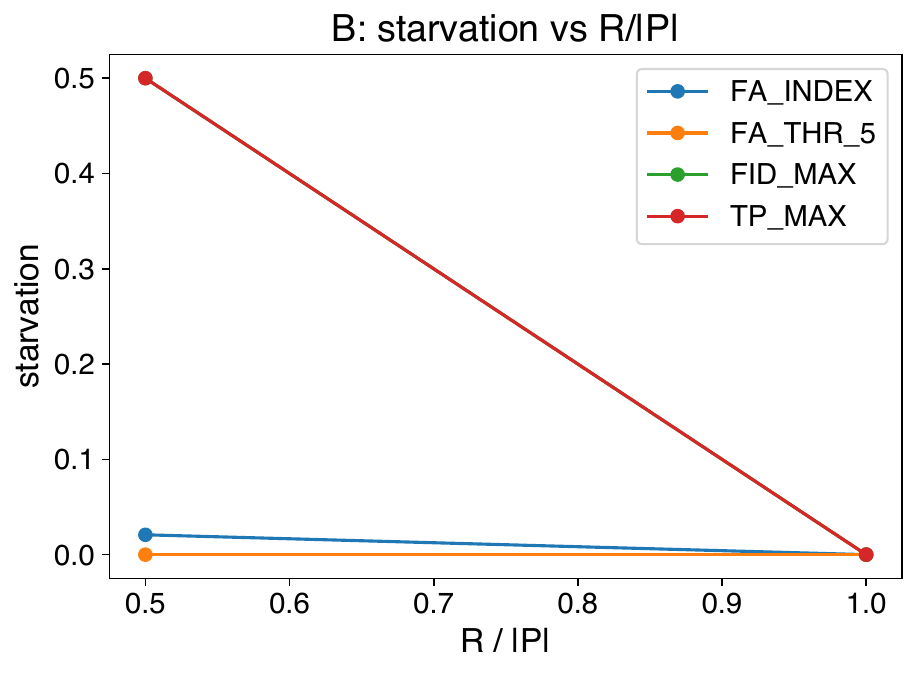}
  \caption{\footnotesize
  Starvation probability $S$ (Eq.~\eqref{eq:Starv}) versus offered-load ratio $R/|P|$, 
  where a flow is deemed starved if its mean age exceeds the single-flow 95th-percentile baseline $A_{\mathrm{ref}}$.
  FA-aware scheduling nearly eliminates starvation even under half-budget load, 
  evidencing strong temporal balancing across flows.}
  \label{fig:B_starv}
\end{subfigure}
\vspace{1ex}
\caption{
$3{\times}3$ grid topology with $|P|{=}16$ flows and attempt-budget ratios $R/|P|{\in}\{0.5,1.0\}$; 
$q{=}0.95$, $F_{\min}{=}0.75$, no purification.
Comparison of fairness and starvation behavior across scheduling policies under varying network load.}
\label{fig:B_dual_fair_starv}
\end{figure}

\subsection{Fidelity Constraints and Purification}
Introducing purification and higher fidelity thresholds changes the trade-off surface between physical reliability and temporal freshness.
As $F_{\min}$ increases, usable entanglement success probability declines, reducing the effective service rate.
Greedy policies amplify this effect because they repeatedly pursue high-fidelity paths that often fail, increasing age variance.
In contrast, FA-aware schedulers naturally adapt: when the fidelity constraint makes a link unreliable, its age term grows, prompting reallocation toward shorter or more probable paths.
Figures~\ref{fig:C_fa95}–\ref{fig:C_thr} reveal that FA-INDEX and FA-THR maintain bounded flow ages while throughput degrades smoothly.
The slope of the FA$_{95}$ curve reflects this adaptive balance—an indicator of resilience against physical-layer tightening.

\begin{figure}[t]
\centering
\begin{subfigure}[t]{0.48\linewidth}
  \centering
  \includegraphics[width=\linewidth]{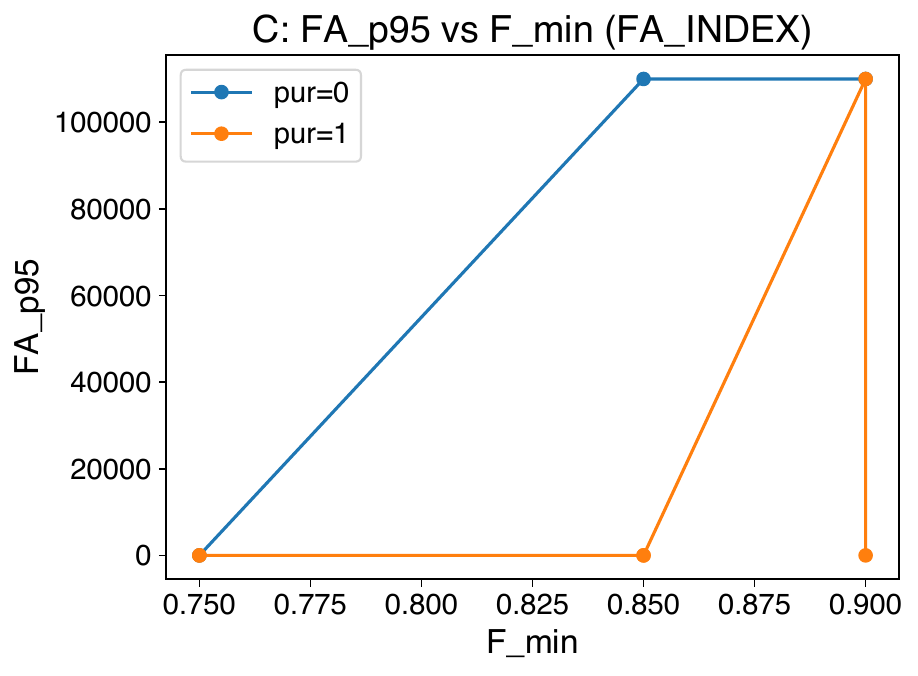}
  \caption{\footnotesize
  95th-percentile flow age $A_{95}$ (Eq.~\eqref{eq:A95}) versus minimum usable-fidelity threshold $F_{\min}$ for the FA-INDEX scheduler.
  Increasing $F_{\min}$ tightens the acceptance criterion and inflates $A_{95}$,
  reflecting longer waiting times for high-fidelity deliveries.}
  \label{fig:C_fa95}
\end{subfigure}
\hfill
\begin{subfigure}[t]{0.48\linewidth}
  \centering
  \includegraphics[width=\linewidth]{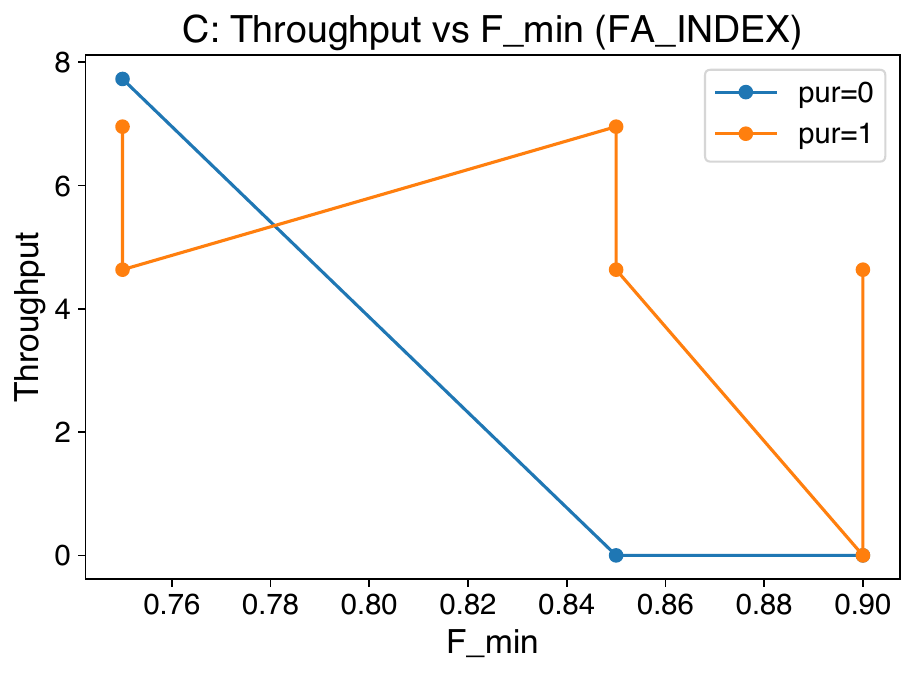}
  \caption{\footnotesize
  Mean throughput $\mathrm{Thr}$ (Eq.~\eqref{eq:Amean}) versus $F_{\min}$.
  Purification improves achievable end-to-end fidelity but lowers the effective delivery rate,
  as more attempts fail to meet stricter thresholds.}
  \label{fig:C_thr}
\end{subfigure}
\vspace{1ex}
\caption{
$3{\times}3$ grid, $|P|{=}8$, attempt-budget ratio $R/|P|{=}1$,
BSM success $q{=}0.95$, purification success $p_{\mathrm{pur}}{=}0.8$, and fidelity increment $\Delta F{=}0.08$.
Comparison of latency and throughput effects as the fidelity threshold $F_{\min}$ varies under FA-INDEX control.}
\label{fig:C_dual_fid}
\end{figure}

\subsection{Heterogeneous Path Lengths}
Path-length heterogeneity introduces asymmetric success probabilities that challenge fairness mechanisms.
Because FA-INDEX currently weighs all flows equally, its age term over-penalizes short paths and under-reacts to long, lossy ones.
The outcome (Fig.~\ref{fig:D1_p95}) is an inflated mean age for distant pairs and fairness loss ($J{\approx}0.64$, Fig.~\ref{fig:D1_jain}).
FA-THR(5), which uses a simpler age-threshold trigger, performs better because it implicitly normalizes by success rate—flows with low link success accumulate age quickly and trigger service more often.
These results highlight that FA-INDEX should incorporate per-path success expectations (e.g., $p_e$ or end-to-end $s_k$) to remain optimal in heterogeneous networks.

\begin{figure}[t]
\centering
\begin{subfigure}[t]{0.48\linewidth}
  \centering
  \includegraphics[width=\linewidth]{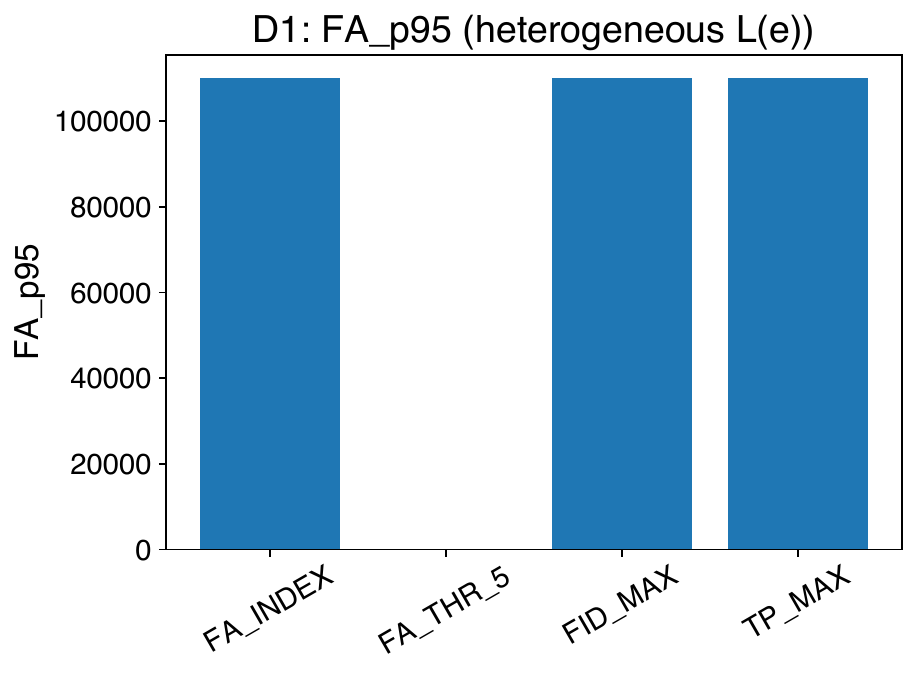}
  \caption{\footnotesize
  95th-percentile flow age $A_{95}$ (Eq.~\eqref{eq:A95}) per policy.
  FA-THR maintains bounded delay across heterogeneous routes,
  whereas FA-INDEX shows higher sensitivity to path length due to stronger age weighting.}
  \label{fig:D1_p95}
\end{subfigure}
\hfill
\begin{subfigure}[t]{0.48\linewidth}
  \centering
  \includegraphics[width=\linewidth]{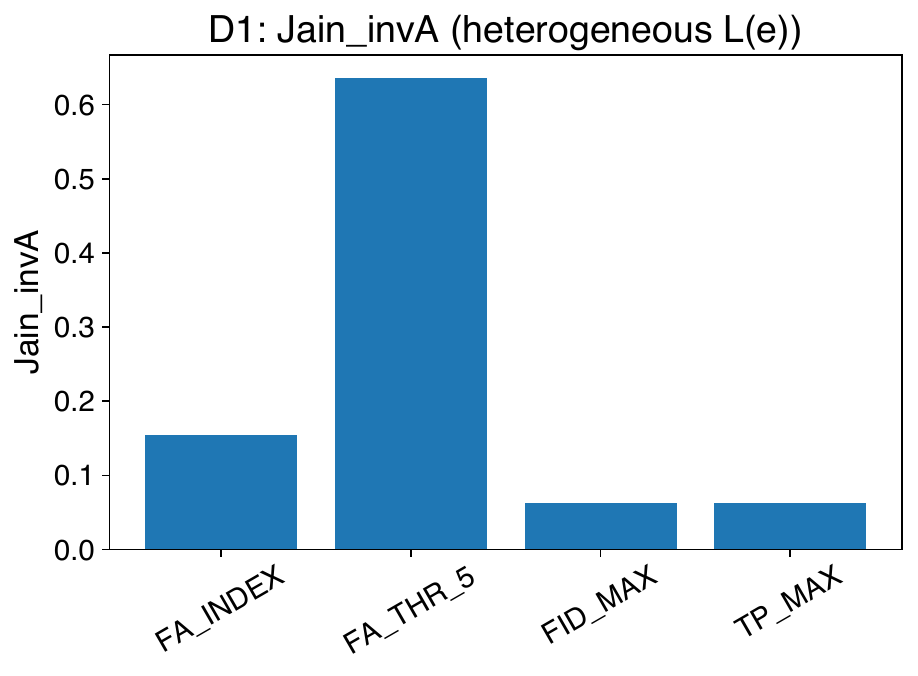}
  \caption{\footnotesize
  Fairness measured by Jain’s index $J$ (Eq.~\eqref{eq:Jain}) computed over inverse mean ages $1/(1+\bar A_{sd})$.
  FA-THR provides the most uniform service across flows,
  while FA-INDEX slightly sacrifices fairness for lower tail age.}
  \label{fig:D1_jain}
\end{subfigure}
\vspace{1ex}
\caption{
$3{\times}3$ grid topology with random link lengths $L(e){\in}[10,80]$ km, producing path lengths $k{\in}\{2,3,4,6\}$;
$|P|{=}8$, $R/|P|{=}0.5$, BSM success $q{=}0.9$, no purification.
Comparison of latency and fairness across policies under heterogeneous end-to-end distances.}
\label{fig:D1_dual}
\end{figure}

\subsection{Finite-Coherence Regime}
As memory coherence degrades (smaller $T_2$ or larger waiting depth $T$), usable fidelities decay exponentially.
This introduces time coupling: delayed deliveries not only increase age but also reduce future success probabilities.
Greedy schedulers ignore this, causing cascading stalls and rapidly increasing $\bar A$ and $A_{95}$.
FA-aware schedulers mitigate this feedback by accelerating service of older flows before fidelity loss becomes severe.
Figure~\ref{fig:D2_fa95} shows that FA-THR maintains bounded ages even when coherence time drops below $100$~ms, while FA-INDEX—tuned for memoryless links—degrades gradually.
Throughput trends in Fig.~\ref{fig:D2_thr} confirm that FA-awareness preserves operational efficiency by anticipating decoherence-induced service decay.

\begin{figure}[t]
\centering
\begin{subfigure}[t]{0.48\linewidth}
  \centering
  \includegraphics[width=\linewidth]{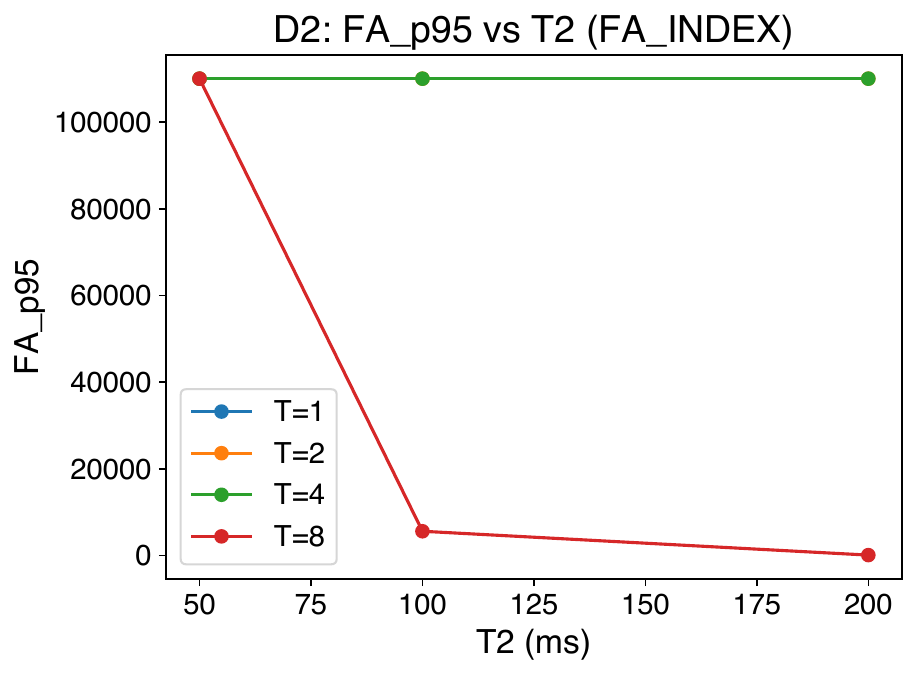}
  \caption{\footnotesize
  95th-percentile flow age $A_{95}$ (Eq.~\eqref{eq:A95}) versus coherence time $T_2$ for the FA-INDEX policy.
  Shorter coherence periods cause rapid latency inflation, especially on long paths, as stored qubits decohere before successful end-to-end completion.}
  \label{fig:D2_fa95}
\end{subfigure}
\hfill
\begin{subfigure}[t]{0.48\linewidth}
  \centering
  \includegraphics[width=\linewidth]{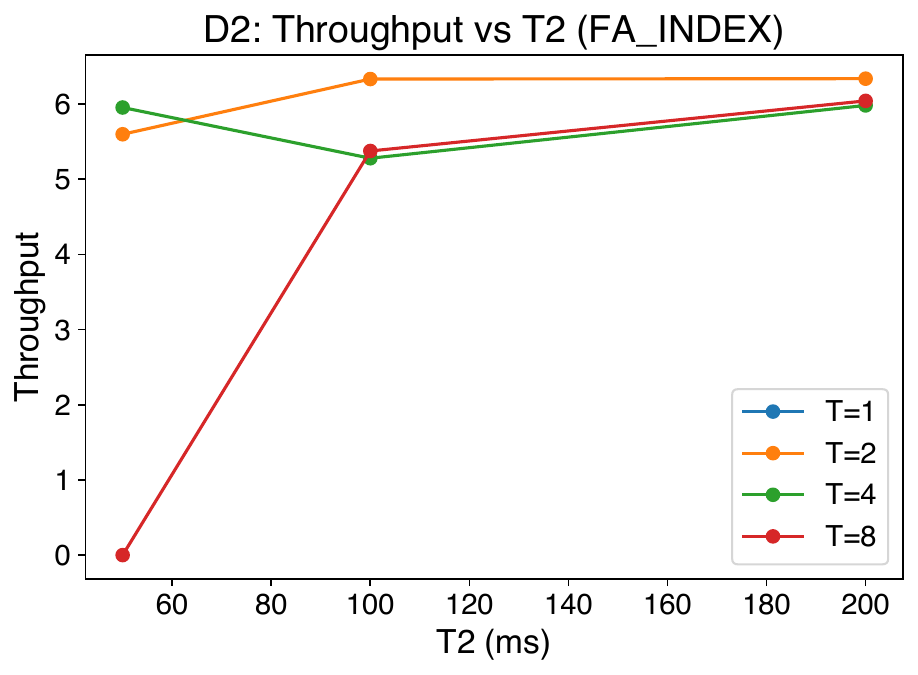}
  \caption{\footnotesize
  Mean throughput $\mathrm{Thr}$ (Eq.~\eqref{eq:Amean}) versus coherence time $T_2$.
  All schedulers suffer rate loss from decoherence,
  but FA-aware control maintains higher throughput by refreshing stale links and prioritizing high-success paths within the coherence window.}
  \label{fig:D2_thr}
\end{subfigure}
\vspace{1ex}
\caption{
$3{\times}3$ grid topology, $|P|{=}8$, attempt-budget ratio $R/|P|{=}1$, 
BSM success $q{=}0.9$, no purification, memory delay $T{\in}\{2,4,8\}$ slots, 
and coherence times $T_{2}{\in}\{50,100,200\}$ ms.
Comparison of latency and throughput degradation under finite-memory effects for FA-INDEX and other schedulers.}
\label{fig:D2_dual}
\end{figure}

\subsection{Pareto Frontier Analysis}
The composite frontier in Fig.~\ref{fig:E_frontier} summarizes the overall trade-off between efficiency (throughput) and temporal freshness (FA$_{95}$).
Points corresponding to FA-THR and FA-INDEX dominate the Pareto boundary across all load regimes, whereas TP-MAX and FID-MAX cluster far from the efficient region.
This demonstrates that enforcing temporal balance via flow-age control yields globally superior operating points—no policy can jointly improve both throughput and latency without incorporating some form of FA-awareness.

\begin{figure}[t]
\centering
\includegraphics[width=0.58\linewidth]{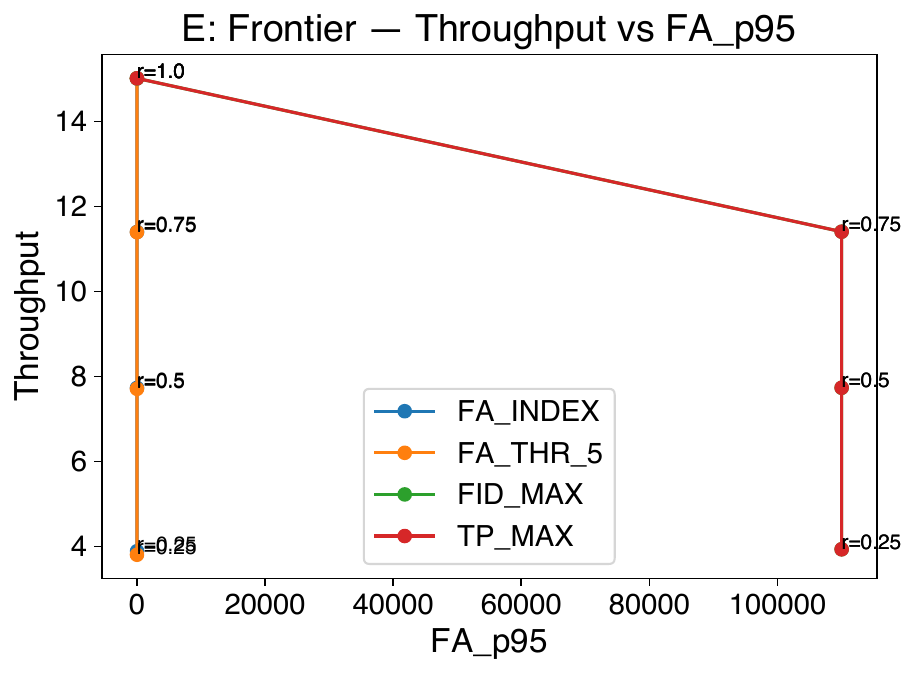}
\caption{
Each point represents the mean throughput $\mathrm{Thr}$ (Eq.~\eqref{eq:Amean}) and the corresponding 95th-percentile flow age $A_{95}$ (Eq.~\eqref{eq:A95}) averaged over 20 Monte Carlo seeds.
Pareto frontiers illustrate the trade-off between delivery rate and latency tail.
FA-aware schedulers trace the Pareto-efficient frontier, maintaining high throughput with markedly smaller tail ages than throughput-driven baselines.}

\label{fig:E_frontier}
\end{figure}

\subsection{Discussion and Insights}

The results show that flow-age weighting serves as a distributed regulator aligning short-term scheduling with long-term temporal stability. 
Purification and fidelity constraints interact multiplicatively with age, requiring joint adaptation of fidelity targets and age weights. 
In heterogeneous or decoherence-limited regimes, FA-INDEX should normalize by estimated success probabilities to preserve fairness. 
Overall, the FA-aware policies behave as dynamic backlog-balancing controllers that approximately minimize age drift, linking physical-layer reliability with network-layer fairness and establishing age-driven scheduling as a unifying principle for stable quantum networks.

\section{Conclusions}

We introduced the \emph{Fidelity-Age (FA)} metric as a joint measure of timeliness and fidelity for quantum repeater networks, capturing the time elapsed since the most recent delivery whose end-to-end fidelity exceeds a prescribed threshold. By modeling usable entanglement deliveries as renewal events, we established an explicit analytical relationship between slot-level success probabilities and long-run average FA, enabling principled performance evaluation under stochastic generation, swapping, and decoherence.

Formulating FA minimization under finite memory, attempt budgets, and fidelity constraints led to two lightweight scheduling policies, \emph{FA-THR} and \emph{FA-INDEX}, which approximate Lyapunov-drift–optimal control while remaining implementable with local state information. These policies directly target freshness-conditioned usability rather than instantaneous throughput alone.

Simulation results demonstrate that FA-aware scheduling preserves throughput while substantially improving fairness and stability, suppressing extreme-age tails even in heterogeneous networks with finite coherence times. More broadly, explicit age-based weighting provides a systematic mechanism for balancing short-term transmission opportunities against long-term delivery regularity, suggesting FA as a unifying objective for freshness-aware control in quantum networks.

%Future work will extend FA control to multi-hop chains with dynamic memory, reinforcement-learning integration, and continuous-time decoherence models.

\balance

\bibliographystyle{IEEEtran}
\bibliography{ref}

\end{document}